\documentclass[10pt]{IEEEtran}

\usepackage{amsmath,amssymb,amsthm}
\usepackage{graphicx}
\usepackage{enumitem}
\usepackage{url}
\usepackage{booktabs}

\newcommand{\real}{\ensuremath{\mathbb{R}}}
\newcommand{\nat}{\ensuremath{\mathbb{N}}}
\newcommand{\one}{\ensuremath{\mathbf{1}}}
\newcommand{\smat}[1]{\ensuremath{\left[\begin{smallmatrix}#1\end{smallmatrix}\right]}}
\newcommand{\bmat}[1]{\ensuremath{\begin{bmatrix}#1\end{bmatrix}}}
\newcommand{\sa}[1]{\mathsf{#1}}
\newcommand{\VVT}{\ensuremath{V_{\mathcal{V}_T}}}

\DeclareMathOperator{\rank}{rank}
\DeclareMathOperator{\ve}{vec}

\newtheorem{definition}{Definition}
\newtheorem{remark}{Remark}
\newtheorem{lemma}{Lemma}
\newtheorem{problem}{Problem}

\newtheorem{theorem}{Theorem}
\newtheorem{proposition}{Proposition}
\newtheorem{fact}{Fact}

\allowdisplaybreaks
\graphicspath{{./img/}}

\title{Controller design for robust invariance\\ from noisy data}
\author{Andrea Bisoffi, Claudio De Persis, Pietro Tesi
\thanks{This research is partially supported by a Marie Sk\l{}odowska-Curie COFUND grant, no.~754315.}
\thanks{A. Bisoffi, C. De Persis are with ENTEG and the J.C. Willems Center for Systems and Control, University of Groningen, 9747 AG Groningen, The Netherlands. P. Tesi is with DINFO, University of Florence, 50139 Florence, Italy. \texttt{\{a.bisoffi, c.de.persis\}@rug.nl, pietro.tesi@unifi.it}}
}

\begin{document}

\maketitle

\begin{abstract}
For an unknown linear system, starting from noisy open-loop input-state data collected during a finite-length experiment,  
we directly design a linear feedback controller that guarantees robust invariance of a given polyhedral set of the state in the presence of disturbances.
The main result is a necessary and sufficient condition for the existence of such a controller,
and amounts to the solution of a linear program. The benefits of large and rich data sets for the solution of the problem are discussed. A numerical example about a simplified platoon of two vehicles illustrates the method.
\end{abstract}

\section{Introduction}
\label{sec:intro}

In many control systems, it is vital to explicitly enforce state constraints. This corresponds to imposing invariance of a set, that is, solutions to the system initialized within the set never leave the set. Due to its practical motivation, invariance has been thoroughly investigated from the late 1980's, and \cite{blanchini1999set,blanchini2008set} are comprehensive guides to these developments. Among these, we focus in this work on the realistic setting when invariance needs to be guaranteed despite a disturbance.

The previous design methods rely on the knowledge of a model of the system, whose parameters need typically to be identified from data. When controlling increasingly complex systems, however, first-principles modeling and system identification may neglect relevant phenomena that are still captured by data.
In such cases, the literature thread of direct data-driven control (see, e.g., \cite{hjalmarsson1998iterative,campi2002virtual,karimi2004iterative}) is a valid approach to design a feedback controller directly from input-output data collected in an experiment, instead of identifying a model of the system through those data as an intermediate step.

With the same motivation of direct data-driven control, we propose here a data-based design of controllers that render a given set robustly invariant.
As a key challenging aspect, the disturbance with respect to which we want to achieve robust invariance, also affects the data, which are then noisy. For the proposed design, (i)~we use input-state data generated by a linear system affected by the (process) disturbance and, as in many model-based solutions,
we consider (ii)~the set of values taken by the disturbance as prior information; (iii)~polyhedral sets for state and disturbance constraints; (iv)~linear state feedback as control policy.

\textbf{Contribution.} Our main contributions are  necessary and sufficient conditions for robust invariance of a polyhedral set of the state that depend only on noisy input-state data, have a simple and explicit expression, and correspond to numerically efficient linear programs. These conditions provide a new data-based solution to a classical invariance problem. 
We also show how a condition akin to persistence of excitation is beneficial to solving the considered data-based robust invariance problem. Moreover, input constraints can be straightforwardly incorporated in our conditions (as indicated in Remark~\ref{remark:input constraints}).%

\textbf{Related work.}
Model-based methods for robust invariance are in \cite{blanchini1990feedback}, and, for the case without disturbance, in \cite{gutman1986admissible,bitsoris1988positive,vassilaki1988feedback,
gilbert1991linear,gravalou1994algorithm}. Some of them \cite{gutman1986admissible,gilbert1991linear} had then a pivotal role in the model predictive literature \cite{bemporad2002explicit}.
Our data-driven approach is inspired by the key insight in 
\cite[Thm.~1]{willems2005note}.
While early use of this insight in the context of data-driven control has appeared, e.g., in~\cite{markovsky2008datadriven,park2009stability}, recent contributions have emphasized its role in the constrained optimal control of unknown stochastic linear systems~\cite{coulson2019data,coulson2020distributionally} and in the derivation of explicit formulas for data-driven control with optimal and robust performance~\cite{depersis2020tac}. 
Since then, many very recent works in data-driven control have appeared, among which \cite{berberich2019robust,vanwaarde2020data,bisoffi2020ifac,monshizadeh2020amidst,vanwaarde2020noisy,
berberich2020robust,guo20cdc,bisoffi20scl}.
In the context of stabilization or $\mathcal{H}_2$, $\mathcal{H}_\infty$ control (instead of robust invariance), data-driven approaches with noisy data were considered in \cite{depersis2020tac,vanwaarde2020noisy,berberich2019robust}. 
Although the quadratic Lyapunov functions considered therein for stabilization may be relied upon to build robustly invariant sublevel sets, we consider here the different setting of polyhedral constraints. The consideration of state constraints relates this work to data-driven predictive control such as \cite{Salvador2018,coulson2019data,berberich2020robust}. 
Unlike these works, our focus is to provide, based on noisy data, necessary and sufficient conditions for when a given polyhedral set can be made robustly invariant.
 With a single offline linear program, we determine if robust invariance is feasible and, if so, yield a controller gain, as opposed to solving an online optimization problem as in predictive control.
We note that in our previous work \cite{bisoffi2020ifac} we considered primarily the nominal case (i.e., no disturbance in the dynamics or in the data) and gave a preliminary sufficient condition for noisy data.
Moreover, we mention \cite{aswani2013provably,garcia2015comprehensive,wabersich2018scalable} for data-driven control with state and input constraints (sometimes termed \emph{safe control}). Our approach considers an unknown linear dynamics with a disturbance as opposed to a known linear dynamics with a disturbance \cite{aswani2013provably} or a known linear dynamics with nonlinear term \cite{wabersich2018scalable}, which is handled by an ellipsoidal underapproximation of the original polyhedral set.
Finally, to solve our problem we use a data-dependent parametrization of the system and Farkas lemma, the latter inspired by analogous results for model-based robust set invariance \cite{vassilaki1988feedback,bitsoris1988positive,blanchini1990feedback}.
After we completed this work, we became aware that these tools have been first proposed in a series of works for the data-driven stabilization of switched systems \cite{Dai2019cdc,Dai2020ifac,Dai2020aut} and nonlinear systems \cite{dai2020semialgebraic}.
Since the two problems addressed here and in the papers above are different, the conditions that guarantee the solution to our problem are distinct from those in \cite{Dai2019cdc,Dai2020aut}.

\textbf{Structure.} Preliminaries on polyhedra, robust invariance, Farkas lemma are in Section~\ref{sec:prel}. In Section~\ref{sec:main}, we present problem formulation and data-based solution, for which we then examine the benefits of ``rich'' data and its connection with the model-based solution. The results are proved in Section~\ref{sec:proofs} and illustrated in Section~\ref{sec:example}.

\section{Notation and preliminaries}
\label{sec:prel}

For vectors $x_1 \in \real^{d_1}$, \dots, $x_m \in \real^{d_m}$, 
the notation $(x_1, \dots,  x_m)$ is equivalent to $[x_1^\top \dots x_m^\top]^\top$.
For a matrix $A=\bmat{a_1 & \dots & a_n}$ partitioned according to its columns, 
define $\ve(A):=\bmat{a_1^\top \dots a_n^\top}^\top$, which returns a vector stacking the columns of $A$ taken from the left to the right.
For matrices $A$ and $B$, $A \otimes B$ denotes their Kronecker product.
For matrices $A \in \real^{m \times n}$, $X \in \real^{n \times p}$, $B \in \real^{p \times q}$, $C \in \real^{m \times q}$, the matrix  equation $A X B = C$ is equivalent (see, e.g., \cite[Lemma~4.3.1]{horn1991topics}) to
\begin{equation}
\label{vec trick}
(B^\top \otimes A) \ve(X)= \ve(A X B)= \ve(C).
\end{equation}
Given a set $\mathcal{A}$ and a scalar $a \ge 0$, $a \mathcal{A}:= \{ a x \colon x \in \mathcal{A} \}$.
For a positive integer $n$, $\nat_n := \{1,\dots,n\}$. 
$\one$ denotes the vector of all ones of appropriate
dimension, i.e., $\one := (1, \dots, 1)$.
$I$ denotes an identity matrix of appropriate
dimension.
$0$ denotes a matrix of all zeros of appropriate dimension.
Given two $n\times m$ matrices $A$ and $B$, $A \ge 0$ indicates that each entry of $A$ is nonnegative,
and $A \ge B$ is equivalent to $A - B \ge 0$.
For a matrix $A \in \real^{m \times n}$ and a vector $b \in \real^m$, a polyhedron is a set of the form $\{ x \in \real^n \colon A x \le b \}$.
Based on the observation in \cite[p.~108]{blanchini2008set}, the next result holds.

\begin{lemma}
\label{lemma:bounded polyhedron}
For $\sa{A} \in \real^{\sa{m} \times \sa{n}}$, $\sa{b}_1 \in \real^\sa{m}$ and $\sa{b}_2 \in \real^\sa{m}$, consider a nonempty polyhedron $\mathcal{A}:=\{ \sa{x} \in \real^\sa{n} \colon \sa{b}_1 \le \sa{A} \sa{x} \le \sa{b}_2 \}$. $\mathcal{A}$ is bounded if and only if $\sa{A}$ has full column rank.
\end{lemma}
\begin{proof}
($\Longrightarrow$)
Suppose by contradiction that $\sa{A}$ has not full column rank.
For $\sa{x} \in \mathcal{A}$ (i.e., $\sa{b}_1 \le \sa{A} \sa{x} \le \sa{b}_2$), also $\sa{x}+\sa{z} \in \mathcal{A}$ where $\sa{z}$ is a nonzero vector in the kernel of $\sa{A}$, since $\sa{A} (\sa{x} + \sa{z} ) = \sa{A} \sa{x}$. Since the norm of $\sa{z}$ can be arbitrarily large, $\sa{x} + \sa{z} \in \mathcal{A}$ contradicts boundedness of $\mathcal{A}$.\newline
($\Longleftarrow$)
For each $\sa{x} \in \mathcal{A}$, let $\sa{y} := \sa{A} \sa{x}$. Since $\sa{A}$ has full column rank, there exists a left inverse $\sa{A}_\textup{l}$ such that $\sa{A}_\textup{l} \sa{A} = I$. Hence,
$
\sa{A}_\textup{l} \sa{y} = \sa{A}_\textup{l} (\sa{A} \sa{x}) = \sa{x}
$.
From this equation, boundedness of $\sa{y}$ (since $\sa{b}_1 \le \sa{y} \le \sa{b}_2$) implies boundedness of $\sa{x}$.
\end{proof}

We consider next the notion of robust invariance.
\begin{definition}\emph{\cite[Def.~2.1]{blanchini1990feedback}}
\label{def:rob inv}
A set $\mathcal{A}$ is \emph{robustly invariant with respect to the set $\mathcal{B}$ for}
\begin{equation}
\label{eq:sysCLdist}
x^+ = F x + d
\end{equation}
if for each initial condition $x(0) \in \mathcal{A}$ and each disturbance $d$ satisfying $d(t) \in \mathcal{B}$ for all $t \ge 0$, the corresponding solution to~\eqref{eq:sysCLdist} satisfies $x(t) \in \mathcal{A}$ for all $t \ge 0$.
\end{definition}

The next fact is a consequence of Farkas lemma credited to R. J. Duffin in \cite[p.~32]{mangasarian1994nonlinear}, and is also the subject of~\cite{hennet1989extension} due to its instrumental role in the invariance literature. Its proof is in the appendix for completeness.
\begin{fact}
\label{fact:farkas}
Let $\sa{A} \in \real^{\sa{p} \times \sa{n}}$, $\sa{B} \in \real^{\sa{q} \times \sa{n}}$, $\sa{c} \in \real^{\sa{p}}$, $\sa{d} \in \real^{\sa{q}}$ and assume there exists $\sa{z} \in \real^{\sa{n}}$ satisfying $\sa{A} \sa{z} \le \sa{c}$. Then, \eqref{farkas-1} and \eqref{farkas-2} are equivalent:
\begin{align}
\label{farkas-1} 
& \sa{B} \sa{x} \le \sa{d} \text{ for all } \sa{x} \in \real^{ \sa{n}} \text{ such that } \sa{A} \sa{x} \le \sa{c};\\
\label{farkas-2}
& \text{there exists } \sa{E} \in \real^{\sa{q} \times \sa{p} }\! \text{ such that } \sa{E} \ge 0, \sa{B} = \sa{E} \sa{A},  \sa{E} \sa{c}\le \sa{d}.
\end{align}
\end{fact}

\section{Data-based guarantees for robust invariance from noisy data}
\label{sec:main}

Consider the discrete-time linear system
\begin{equation}
\label{sys}
x^+ = A x + B u + d
\end{equation}
where $x \in \mathbb{R}^n$ is the state, $u \in \mathbb{R}^m$ is the input, $d \in \mathbb{R}^n$ is the disturbance, and the system matrices have dimensions $A\in \mathbb{R}^{n\times n}$, $B\in \mathbb{R}^{n \times m}$.

For given matrices $\mathrm{S} \in \real^{n_s \times n}$ and $\mathrm{D} \in \real^{n_d \times n}$ and scalar parameter $\delta \ge 0$, we consider the polyhedral sets $\mathcal{S}$ and $\mathcal{D}_\delta$
\begin{align}
\mathcal{S} := & \{ x \in \real^n  \colon \mathrm{S} x \le \one \} \label{set S}\\
\mathcal{D}_\delta := & \{ d \in \real^n  \colon \mathrm{D} d \le \delta \one \}. \label{set D}
\end{align}
We make $\delta$ explicit as a subscript to account for disturbances with different magnitudes, where a smaller $\delta$ ($\delta=0$ in the limit) corresponds to a smaller set $\mathcal{D}_\delta$.

\begin{remark}[C-set]
The set $\mathcal{S}$ is convex, closed, and includes the origin as interior point; if we further assume that it is bounded, $\mathcal{S}$ is a C-set \cite[Def.~3.10]{blanchini2008set}.
\end{remark}

\subsection{The model-based problem formulation and solution}
\label{sec:model-based}

To motivate the data-based problem formulation in Problem~\ref{probl:data-based} below and the corresponding solution in Section~\ref{sec:data-based solution}, we first point out the problem and its solution in a model-based setting.

\begin{problem}
\label{probl:model-based}
For the given $\mathcal{S}$ and $\mathcal{D}_\delta$, design a linear state-feedback controller $u=K x$ for~\eqref{sys} so that 
\begin{equation}
\label{MB-robInv}
(A + BK) x +d \in \mathcal{S} \text{ for all } x \in \mathcal{S} \text{ and } d \in \mathcal{D}_\delta.
\end{equation}
\end{problem}
We note that the condition in~\eqref{MB-robInv} is equivalent to having $\mathcal{S}$ in~\eqref{set S} robustly invariant with respect to $\mathcal{D}_\delta$ in~\eqref{set D} for the closed-loop system $x^+ = (A + B K) x + d$ \cite[Proof of Thm.~2.1]{blanchini1990feedback}.
By \eqref{set S} and \eqref{set D}, \eqref{MB-robInv} is equivalent to
\begin{equation*}
\begin{split}
& \begin{bmatrix}
\mathrm{S}(A+BK) & \mathrm{S}
\end{bmatrix}
\begin{bmatrix}
x\\
d
\end{bmatrix}
\le \one \\
& \hspace*{5mm} \text{ for all }
\begin{bmatrix}
x\\
d
\end{bmatrix}
\text{ such that }
\begin{bmatrix}
\mathrm{S} & 0\\
0 & \mathrm{D}
\end{bmatrix}
\begin{bmatrix}
x\\
d
\end{bmatrix}
\le 
\bmat{
\one\\
\delta \one
}.
\end{split}
\end{equation*}
By Fact~\ref{fact:farkas} (whose assumption is verified by $(x,d)=0$), this is equivalent to
\begin{equation}
\label{MB-robInv-2}
\begin{split}
& \text{there exists } E\in \real^{n_s \times (n_s + n_d)} \text{ such that } \\
& \hspace*{3mm}  E \ge 0,\begin{bmatrix}
\mathrm{S}(A+BK) & \mathrm{S}
\end{bmatrix} = E \begin{bmatrix}
\mathrm{S} & 0\\
0 & \mathrm{D}
\end{bmatrix},  
E
\bmat{
\one\\
\delta \one
} \le \one,
\end{split}
\end{equation}
which is a linear program in the decision variables $E$ and $K$. Analogous results were obtained originally in, e.g., \cite{vassilaki1988feedback,bitsoris1988positive,blanchini1990feedback} within the invariance literature.

\subsection{The data-based problem formulation}

Our aim is solving a robust invariance problem only through a finite number of input and state data, instead of using the model knowledge embedded in matrices $A$ and $B$ as in Section~\ref{sec:model-based}.

To this end, in an open-loop experiment we apply $T$ input data $u(0), u(1), \dots, u(T-1)$, and measure the state data $x(0), x(1), \dots, x(T)$. We assume that this state response is generated according to~\eqref{sys}, so it contains also the effect of the disturbance samples $d(0), d(1), \dots, d(T-1)$. These disturbance samples are \emph{not} known to us, except for the fact that each of them belongs to the \emph{known} set $\mathcal{D}_\delta$.

We organize the input data, state data, and disturbance samples as
\begin{subequations}
\label{data}
\begin{align}
U_{0,T} & := 
\begin{bmatrix}
u(0) & u(1) & \ldots & u(T-1)
\end{bmatrix} \label{U0}\\
X_{0,T} & := 
\begin{bmatrix}
x(0) & x(1) & \ldots & x(T-1)
\end{bmatrix} \label{X0}\\
X_{1,T} & := 
\begin{bmatrix} 
x(1) & x(2) & \ldots & x(T)
\end{bmatrix} \label{X1}\\
D_{0,T} & :=
\begin{bmatrix}
d(0) & d(1) & \ldots & d(T-1)
\end{bmatrix}, \label{D0}
\end{align}
\end{subequations}
based on which we define
\begin{equation}
W_{0,T} := \bmat{X_{0,T}\\ U_{0,T} } \in \real^{(n+m)\times T}. \label{W0}
\end{equation}
The data in~\eqref{data}, generated according to~\eqref{sys}, satisfy
\begin{equation}
\label{data-relation}
X_{1,T}= A X_{0,T} + B U_{0,T} + D_{0,T}.
\end{equation}
We note that $X_{1,T}$ and $X_{0,T}$ are affected by $D_{0,T}$, so the data are \emph{noisy}.

Our objective is then to find a controller that guarantees robust invariance despite the imprecise characterization of the system obtained through noisy data. Specifically, noisy data allow us to characterize all \emph{consistent} matrices $(\hat A,\hat B)$, i.e., the matrices that are compatible with the measured data and the bound on the disturbance samples captured by the set $\mathcal{D}_\delta$.
More concretely, the consistent matrices $(\hat A,\hat B)$ must satisfy 
\begin{equation}
X_{1,T}=\hat A X_{0,T} + \hat B U_{0,T} + \hat D
\end{equation}
for some matrix $\hat D$, each of whose columns must belong to $\mathcal{D}_\delta$. These conditions are captured by the next set $\mathcal{V}_T$, where we make $T$ explicit as a subscript and $(M)_i$ denotes the $i$-th column of a matrix $M$:
\begin{equation}
\label{cal V}
\begin{split}
&\mathcal{V}_T:= \{\bmat{\hat A & \hat B} \colon X_{1,T} = \hat A X_{0,T} + \hat B U_{0,T} + \hat D,\\
& \hspace*{3.5cm} (\hat D)_i \in \mathcal{D}_\delta \text{ for all } i \in \nat_T\}.
\end{split}
\end{equation}
We use $W_{0,T}$ in~\eqref{W0} and write more compactly the set $\mathcal{V}_T$  as
\begin{equation}
\label{cal V-1}
\begin{split}
& \mathcal{V}_T= \{ V \in \real^{n \times (n+m)} \colon \\
& \hspace*{2cm}(X_{1,T} - V W_{0,T})_i \in \mathcal{D}_\delta \text{ for all } i \in \nat_T \}.
\end{split}
\end{equation}
Note that \eqref{cal V-1} is completely determined by the data $X_{0,T}$, $U_{0,T}$ (in $W_{0,T}$), $X_{1,T}$. Through $V \in \mathcal{V}_T$ (hence, through data $X_{0,T}$, $U_{0,T}$, $X_{1,T}$), the closed loop of any consistent matrices $(\hat A,\hat B)$ with feedback $u= Kx$ is
\begin{equation*}
\begin{split}
 x^+ & = \hat A x + \hat B (Kx) + d \\
& = \bmat{\hat A & \hat B}\bmat{I\\ K} x + d = V \bmat{I\\ K} x + d.
\end{split}
\end{equation*}
This discussion allows us to fully formulate the problem statement in our data-based setting.
\begin{problem}
\label{probl:data-based}
For the given $\mathcal{S}$ and $\mathcal{D}_\delta$, design a linear state-feedback controller $u=K x$ only through the input and state data $X_{0,T}$, $U_{0,T}$, $X_{1,T}$ so that
\begin{equation}
\label{rob-inv}
\begin{split}
V \bmat{I\\ K} x + d\in \mathcal{S} \text{ for all } x \in \mathcal{S}, d \in \mathcal{D}_\delta \text{ and } V \in \mathcal{V}_T.
\end{split}
\end{equation}
\end{problem}
Problem~\ref{probl:data-based} is the counterpart of Problem~\ref{probl:model-based} where \eqref{rob-inv} aims at designing $K$ robustly with respect to the set $\mathcal{V}_T$ induced by the lack of exact knowledge of the model due to the disturbance samples $d(0)$, \dots, $d(T-1)$.

\subsection{The data-based solution}
\label{sec:data-based solution}

We propose here the solution to Problem~\ref{probl:data-based}. 
When $\mathcal{S}$ in~\eqref{set S} is bounded, it can be written \cite[pp.~107-108]{blanchini2008set} as
\begin{equation}
\label{set S - vertices}
\begin{split}
\mathcal{S} = & \Bigg\{ \sum_{j=1}^{V_\mathcal{S}} \alpha_j x^j \colon \one^\top \alpha =1, \alpha \ge 0 \Bigg\}
\end{split}
\end{equation}
where $x^1, \dots, x^{V_\mathcal{S}}$ are the $V_\mathcal{S}$ vertices  of $\mathcal{S}$.
Our main result in Theorem~\ref{thm:main} gives a necessary and sufficient condition for the solution of Problem~\ref{probl:data-based} in terms of noisy data. 
The proof is in Section~\ref{sec:proof rob inv data-based}.
\begin{theorem}
\label{thm:main}
Let $\mathcal{S}$ in~\eqref{set S} be bounded. The next two statements are equivalent.
\begin{enumerate}[leftmargin=12pt]
\item \label{thm:main:statement 1} There exists $K \in \real^{m \times n}$ such that \eqref{rob-inv} holds. 
\item \label{thm:main:statement 2} There exist $K\in \real^{m \times n}$ and $E^j \in \real^{n_s \times (n_d + T n_d)}$ with $j \in \nat_{V_\mathcal{S}}$ such that
\begingroup
\setlength{\arraycolsep}{2.pt}
\begin{align}
\hspace*{-1mm}& E^j \ge 0,  \bmat{\mathrm{S} & \,\bigg(\bmat{I \\ K} x^j \bigg)^\top\! \otimes \mathrm{S} }
= E^j \bmat{
\mathrm{D} & 0 \\
0& - (W_{0,T}^\top \otimes \mathrm{D}) }, \notag \\
\hspace*{-1mm}& E^j \bmat{\delta \one\\ \delta \one - \mathrm{D} x(1)\\ \vdots \\ \delta \one - \mathrm{D} x(T)} \le \one.
\label{main:without x and dist BIS}
\end{align}%
\endgroup
\end{enumerate}%
\end{theorem}

As a nice feature, \eqref{main:without x and dist BIS} corresponds to a linear program.
Moreover, it is straightforward to incorporate input constraints in this data-based linear program, as shown in the next remark.

\begin{remark}[Input constraints]
\label{remark:input constraints}
In addition to robust invariance, consider input constraints that are expressed, for a given matrices $\mathrm{U} \in \real^{n_u \times m}$, by the polyhedral set
\begin{equation}
\mathcal{U} := \{ u \in \real^m  \colon \mathrm{U} u \le \one \}. \label{set U}
\end{equation}
These input constraints, which ask $K x \in \mathcal{U}$ for all $x \in \mathcal{S}$, amount then to
\begin{equation*}
\mathrm{U} K x \le \one \text{ for all } \mathrm{S} x \le \one,
\end{equation*}
which is equivalent by Fact~\ref{fact:farkas} (since its assumption is verified) to the existence of $E_u \in \real^{n_u \times n_s}$ such that
\begin{equation}
E_u \ge 0, \mathrm{U} K = E_u \mathrm{S}, E_u \one \le \one.
\end{equation}
This condition is linear in the decision variables $E_u$ and $K$, and can be readily added to the linear program of Theorem~\ref{thm:main}.
\end{remark}

\subsection{Benefits of large and rich data sets}
\label{sec:benefits large and rich}

So far no assumptions were made on the data used to solve the linear program in~\eqref{main:without x and dist BIS}. However, data sets that are larger and carry more information are intuitively beneficial to the design of $K$ through~\eqref{main:without x and dist BIS}, as we show in the rest of the section.

A first straightforward observation regards large data sets.
Whenever we add an extra data point $u(T)$, $x(T+1)$ to the data $U_{0,T}$, $X_{0,T}$ and $X_{1,T}$, this leads from~\eqref{cal V-1} to the set
\begin{equation*}
\begin{split}
& \mathcal{V}_{T+1}= \{ V \in \real^{n \times (n+m)} \colon \\
& \hspace*{5mm} (X_{1,T} - V W_{0,T})_i \in \mathcal{D}_\delta \text{ for all } i \in \nat_T,\\
& \hspace*{24mm} x(T+1) - V \bmat{x(T)\\ u(T)} \in \mathcal{D}_\delta \} \subseteq \mathcal{V}_T,
\end{split}
\end{equation*}
since the additional data point simply corresponds to an additional constraint in $\mathcal{V}_{T+1}$ as compared to $\mathcal{V}_T$.
Hence, the set $\mathcal{V}_T$ remains the same or becomes smaller as the number $T$ of data increases, and this makes in turn easier to find a controller $K$ satisfying \eqref{rob-inv} in Problem~\ref{probl:data-based}.

A second observation regards ``richness'' of data. 
Proposition~\ref{prop:bounded cal V_T} below shows that for ``rich'' data, the set $\mathcal{V}_T$ of consistent dynamical matrices  is
guaranteed to be bounded. The notion of ``richness'' is associated here with the full row rank of the matrix $W_{0,T}$ in~\eqref{W0}.
Indeed, the full row rank of $W_{0,T}$ is related to considering persistently exciting inputs by, e.g., \cite[Cor.~2]{willems2005note}, which claims that for a controllable pair $(A,B)$ of $x^+ = A x + B u$, if the input sequence $u(0)$, \dots, $u(T-1)$ is persistently exciting of order $n+1$, then $W_{0,T}$ in~\eqref{W0} has full row rank. We now illustrate the relevance of the full-row-rank condition on $W_{0,T}$ even in the considered case with disturbance.
To do so, we replace, within the rest of \emph{this} section, the set $\mathcal{D}_\delta$ (for $\delta \ge 0$ as before) with
\begin{equation}
\mathcal{\hat D}_\delta := \{ d \in \real^n  \colon \delta d_\textup{l} \le \mathrm{\hat D} d \le \delta d_\textup{u} \} \label{set hatD}
\end{equation}
where $\mathrm{\hat D} \in \real^{\hat{n}_d \times n}$ and 
$d_\textup{l} < 0 < d_\textup{u}$, i.e., $0$ belongs to the interior of $\mathcal{\hat D}_\delta$ for $\delta>0$. We restrict ourselves to $\mathcal{\hat D}_\delta$ for technical convenience so that we can apply to it Lemma~\ref{lemma:bounded polyhedron} (whereas the analogue of Lemma~\ref{lemma:bounded polyhedron} for $\mathcal{D}_\delta$ in~\cite[p.~119, Ex. 11]{blanchini2008set} would generate more involved conditions).

\begin{remark}[$\mathcal{\hat D}_\delta$ and $\mathcal{D}_\delta$]
\eqref{set hatD} is rewritten immediately as
\begin{equation*}
\mathcal{\hat D}_\delta = \left\{ d \in \real^n  \colon \bmat{\mathrm{\hat D}\\ - \mathrm{\hat D}} d \le \bmat{\delta d_\textup{u}\\ - \delta d_\textup{l}} \right\} = \{ d \in \real^n  \colon \mathrm{\tilde D} d \le \delta \one \} 
\end{equation*}
for some matrix $\mathrm{\tilde D}$, since $d_\textup{u} > 0$ and  $- d_\textup{l} > 0$. 
This expression shows that $\mathcal{\hat D}_\delta$ is a special case of $\mathcal{D}_\delta$ in~\eqref{set D} for $\mathrm{D}= \mathrm{\tilde D}$.
\end{remark}

With the set $\mathcal{\hat D}_\delta$, we can characterize exactly when the set $\mathcal{V}_T$ is bounded in the next proposition, whose
proof is in Section~\ref{sec:proof prop bounded cal V_T}.

\begin{proposition}
\label{prop:bounded cal V_T}
$\mathcal{V}_T$ in~\eqref{cal V-1} (with $\mathcal{D}_\delta$ replaced by $\mathcal{\hat D}_\delta$) is a bounded polyhedron if and only if $W_{0,T}$ in~\eqref{W0} has full row rank and $\mathrm{\hat D}$ in~\eqref{set hatD} has full column rank.
\end{proposition}

Proposition~\ref{prop:bounded cal V_T} establishes then that ``rich'' data, associated with full row rank of $W_{0,T}$, 
make easier to find $K$ satisfying \eqref{rob-inv} since, again, they shrink the uncertainty set for $V$.
Full row rank of $W_{0,T}$ can be checked from data.

\subsection{An alternative linear program formulation}

Finally, we propose a linear program formulation alternative to Theorem~\ref{thm:main}.

When $\mathcal{V}_T$ in~\eqref{cal V-1} is bounded, each element $V \in \mathcal{V}_T$ can be written as the convex hull of the (matrix) vertices $V^j$ of $\mathcal{V}_T$ with $j=1, \dots, \VVT$, as was done for $\mathcal{S}$ in~\eqref{set S - vertices}. 

This leads to the alternative linear program in the next theorem, whose proof is in Section~\ref{sec:proof thm alternative linear program}.

\begin{theorem}
\label{thm:alternative linear program}
Let $\mathcal{V}_T$ in~\eqref{cal V-1} be bounded. The next two statements are equivalent.
\begin{enumerate}[leftmargin=12pt]
\item There exists $K \in \real^{m \times n}$ such that \eqref{rob-inv} holds. 
\item There exist $K\in \real^{m \times n}$ and $E^j \in \real^{n_s \times (n_s + n_d)}$ with $j \in \nat_{\VVT}$ such that
\begingroup
\setlength{\arraycolsep}{3.5pt}
\begin{equation}
\label{alternative linear program}
E^j\! \ge 0, \bmat{\mathrm{S}V^j \bmat{I\\ K}\, & \mathrm{S}} \!= E^j\! \bmat{\mathrm{S} & 0 \\ 0 & \mathrm{D}}, E^j \! \bmat{\one \\ \delta \one} \le \one.
\end{equation}%
\endgroup
\end{enumerate}%
\end{theorem}

Some comments on Theorem~\ref{thm:alternative linear program} follow. 
First, the comparison of~\eqref{alternative linear program} with \eqref{MB-robInv-2} shows that, for a bounded $\mathcal{V}_T$, the solution of Theorem~\ref{thm:alternative linear program} is the natural data-based extension of the model-based solution. 
Second, whereas Theorem~\ref{thm:main} operates under boundedness of $\mathcal{S}$, Theorem~\ref{thm:alternative linear program} operates under boundedness of $\mathcal{V}_T$, which can be checked by Proposition~\ref{prop:bounded cal V_T}. 
Third, it is easier to find the vertex representation of $\mathcal{S}$ than that of $\mathcal{V}_T$ since $\mathcal{S}$ is characterized by fewer variables than $\mathcal{V}_T$ ($n$ instead of $n(n+m)$) and the halfspace representation of $\mathcal{S}$ typically involves significantly fewer constraints than that of $\mathcal{V}_T$ ($n_s$ instead of $n_d T$, where $n_s$ and $n_d$ are comparable), which may result in a larger number of vertices to compute for $\mathcal{V}_T$.
For this reason, while Theorem~\ref{thm:alternative linear program} gives a natural data-based
counterpart of the model-based solution, Theorem~\ref{thm:main} gives
a numerically more appealing, yet equivalent, data-based solution.

\section{Proofs}
\label{sec:proofs}

\subsection{Proof of Theorem~\ref{thm:main}}
\label{sec:proof rob inv data-based}

We first get rid of $d$ and $V$ in~\eqref{rob-inv}, and then of $x$. The first step is accomplished in the next lemma.
\begin{lemma} 
\label{lemma:rid of the dist-3}
\eqref{rob-inv} holds if and only if for all $x \in \mathcal{S}$, 
there exists $E \in \real^{n_s \times (n_d +T n_d)}$ such that
\begin{equation}
\label{sol3:rid of the dist}
\begin{split}
& E \ge 0, \\
& \bmat{\mathrm{S} & \bigg(\bmat{I \\ K} x \bigg)^\top \otimes \mathrm{S} }
= E \bmat{
\mathrm{D} & 0 \\
0& - (W_{0,T}^\top \otimes \mathrm{D}) },\\
& E \bmat{\delta \one\\ \delta \one - \mathrm{D} x(1)\\ \vdots \\ \delta \one - \mathrm{D} x(T)} \le \one.
\end{split}
\end{equation}
\end{lemma}%
\begin{proof}
\eqref{rob-inv} holds if and only if for all $x \in \mathcal{S}$,
\begin{equation}
\label{rob-inv-for-all-dist}
\mathrm{S} V \bmat{I\\ K} x + \mathrm{S} d \le \one \text{ for all } d \in \mathcal{D}_\delta \text{ and } V \in \mathcal{V}_T.
\end{equation}
In order to apply Fact~\ref{fact:farkas}, we can suitably rewrite this expression in terms of $\ve(V)$, instead of $V$. First, we have by~\eqref{vec trick} that
\begin{equation}
\label{stack V-1}
\mathrm{S} V \bmat{I\\ K} x  = \bigg( \bigg(\bmat{I \\ K} x \bigg)^\top \otimes \mathrm{S} \bigg) \ve(V).
\end{equation}
Then, recall that $V \in \mathcal{V}_T$ amounts from~\eqref{cal V-1} to having $\mathrm{D}(X_{1,T} - V W_{0,T})_i \le \delta \one$ for all $i \in \nat_T$. By stacking  these constraints and applying \eqref{vec trick}, $V \in \mathcal{V}_T$ is equivalent to
\begin{equation}
\label{stack V-2}
- (W_{0,T}^\top \otimes \mathrm{D}) \ve(V) \le
\bmat{\delta \one - \mathrm{D} x(1)\\ \vdots \\ \delta \one - \mathrm{D} x(T)}.
\end{equation}
Hence, by~\eqref{stack V-1}, \eqref{stack V-2} and recalling that $d \in \mathcal{D}_\delta$ is equivalent to $\mathrm{D} d \le \delta \one$, we conclude that \eqref{rob-inv-for-all-dist} is equivalent to
\begin{equation}
\label{before applying farkas}
\begin{split}
& \bmat{\mathrm{S} & \bigg(\bmat{I \\ K} x \bigg)^\top \otimes \mathrm{S} } \bmat{d\\ \ve(V)} \le \one\\ 
& \hspace*{2.5mm} \text{ for all } \bmat{d\\ \ve(V)} \text{ such that }  \\
& \hspace*{4.5mm}
\bmat{
\mathrm{D} & 0 \\
0& - (W_{0,T}^\top \otimes \mathrm{D}) } \bmat{d\\ \ve(V)} \le 
\bmat{\delta \one\\ \delta \one - \mathrm{D} x(1)\\ \vdots \\ \delta \one - \mathrm{D} x(T)}.
\end{split}
\end{equation}
Let us show that the assumption of Fact~\ref{fact:farkas} is verified, so that we can apply it to~\eqref{before applying farkas}. This assumption is verified if there exists $(d,\ve(V))$ satisfying the last inequality in~\eqref{before applying farkas}, which is equivalent to the existence of $d \in \mathcal{D}_\delta$ and $V \in \mathcal{V}_T$.
Since data are generated according to~\eqref{sys} and satisfy \eqref{data-relation}, $\bar V = \bmat{A & B}$ satisfies 
\begin{equation}
\label{verify assumption farkas}
(X_{1,T} - \bar V W_{0,T})_i = (D_{0,T})_i \in \mathcal{D}_\delta \text{ for all } i \in \nat_T,
\end{equation}
hence $\bar V \in \mathcal{V}_T$. With $d=0 \in \mathcal{D}_\delta$, this concludes the existence of $d \in \mathcal{D}_\delta$ and $V \in \mathcal{V}_T$, so the assumption of Fact~\ref{fact:farkas} is verified. Then, we apply Fact~\ref{fact:farkas} with clear correspondences between the quantities in~\eqref{before applying farkas} and those in~\eqref{farkas-1}, and obtain that \eqref{rob-inv} holds if and only if for all $x \in \mathcal{S}$, there exists $E \in \real^{n_s \times (n_d +T n_d)}$ such that \eqref{sol3:rid of the dist} holds.
\end{proof}

\begin{remark}[Redundant constraints]
The characterization of $\mathcal{V}_T$ given by~\eqref{stack V-2} may present redundant constraints. Removing those constraints and obtaining a \emph{minimal} halfspace representation of the polyhedron (see \cite[p.~48]{ziegler2012lectures}) can be beneficial to having a decision variable $E$ with fewer columns than $n_d + n_d T$.
\end{remark}

We emphasize that the quantity $E$ in Lemma~\ref{lemma:rid of the dist-3} can depend on $x$, even though this dependence is not explicit. 
This observation is essential for the next step (see in particular the proof of Lemma~\ref{lemma:vertices}). By using the expression~\eqref{set S - vertices} in terms of the $V_\mathcal{S}$ vertices $x^1, \dots, x^{V_\mathcal{S}}$ and checking the condition in Lemma~\ref{lemma:rid of the dist-3} only at these vertices, we can restate Lemma~\ref{lemma:rid of the dist-3} as in the next lemma without introducing conservatism.

\begin{lemma}
\label{lemma:vertices}
Let $\mathcal{S}$ in~\eqref{set S} be bounded. \eqref{rob-inv} holds if and only if there exists 
$E^j \in \real^{n_s \times (n_d + T n_d)}$ with $j \in \nat_{V_\mathcal{S}}$ such that \eqref{main:without x and dist BIS} holds.
\end{lemma}
\begin{proof}
The $\!\!\implies\!\!$-direction is trivial by using Lemma~\ref{lemma:rid of the dist-3} for $x=x^j \in \mathcal{S}$ with $j \in \nat_{V_\mathcal{S}}$.
The $\!\!\impliedby\!\!$-direction is proven if for $E^j$ with $j \in \nat_{V_\mathcal{S}}$ satisfying~\eqref{main:without x and dist BIS} and an arbitrary $x \in \mathcal{S}$, we find $E$ that can depend on $x$ and satisfies~\eqref{sol3:rid of the dist}. 
This is done by exploiting that for a bounded $\mathcal{S}$ as in~\eqref{set S - vertices}, an arbitrary $x \in \mathcal{S}$ can always be written in terms of a convex combination of the vertices of $\mathcal{S}$ as
$
x=\sum_{j=1}^{V_\mathcal{S}} \alpha_j x^j 
$
through coefficients $\alpha_j$ ($j \in \nat_{V_\mathcal{S}}$) that satisfy $\one^\top \alpha =1$ and $\alpha \ge 0$. By using these same coefficients, take
\begin{equation}
\label{sol3:selection E}
E = \sum_{j=1}^{V_\mathcal{S}} \alpha_j E^j,
\end{equation}
and let us show that such $E$ satisfies \eqref{sol3:rid of the dist} to complete the proof. 
Since $\alpha \ge 0$ and $E^j \ge 0$ for each $j \in \nat_{V_\mathcal{S}}$ by~\eqref{main:without x and dist BIS}, we have from~\eqref{sol3:selection E}
\begin{subequations}
\begin{equation}
\label{sol3:E>=0}
E \ge 0.
\end{equation}
By $\one^\top \alpha=1$ and the properties of the Kronecker product, we have
\begin{align}
& \bmat{\mathrm{S} & \bigg(\bmat{I \\ K} x \bigg)^\top \otimes \mathrm{S} }
= \bmat{\mathrm{S} & \bigg(\bmat{I \\ K} \sum_{j=1}^{V_\mathcal{S}} \alpha_j x^j \bigg)^\top \otimes \mathrm{S} } \notag\\
& = \bmat{\sum_{j=1}^{V_\mathcal{S}} \alpha_j \mathrm{S} &  \sum_{j=1}^{V_\mathcal{S}} \alpha_j \bigg( \big( \bmat{I \\ K}x^j\big)^\top  \otimes \mathrm{S} \bigg)} \\
& \overset{\eqref{main:without x and dist BIS}}{=} \sum_{j=1}^{V_\mathcal{S}} \alpha_j E^j \bmat{
\mathrm{D} & 0 \notag \\
0& - (W_{0,T}^\top \otimes \mathrm{D}) } \overset{\eqref{sol3:selection E}}{=} E \bmat{
\mathrm{D} & 0 \\
0& - (W_{0,T}^\top \otimes \mathrm{D}) }. \notag
\end{align}
By $\one^\top \alpha=1$ and $\alpha \ge 0$, we have
\begin{equation}
\label{sol3:E vec le}
\begin{split}
& E \bmat{\delta \one\\ \delta \one - \mathrm{D} x(1)\\ \vdots \\ \delta \one - \mathrm{D} x(T)} 
\overset{\eqref{sol3:selection E}}{=} \sum_{j=1}^{V_\mathcal{S}} \alpha_j E^j \bmat{\delta \one\\ \delta \one - \mathrm{D} x(1)\\ \vdots \\ \delta \one - \mathrm{D} x(T)} \\
& 
\overset{\eqref{main:without x and dist BIS}}{\le}
\sum_{j=1}^{V_\mathcal{S}} \alpha_j \one = \one.
\end{split}
\end{equation}
\eqref{sol3:E>=0}--\eqref{sol3:E vec le} show that \eqref{sol3:rid of the dist} is satisfied and the proof is complete.
\end{subequations}
\end{proof}

Lemma~\ref{lemma:vertices} shows indeed that Theorem~\ref{thm:main} holds.

\subsection{Proof of Proposition~\ref{prop:bounded cal V_T}}
\label{sec:proof prop bounded cal V_T}

With the set $\mathcal{\hat D}_\delta$ in~\eqref{set hatD}, the set $\mathcal{V}_T$ in~\eqref{cal V-1} becomes
\begin{equation}
\label{cal V-2}
\begin{split}
& \mathcal{V}_T= \{ V \in \real^{n \times (n+m)}\colon \\
& \hspace*{2cm}(X_{1,T} - V W_{0,T})_i \in \mathcal{\hat D}_\delta \text{ for all } i \in \nat_T \}.
\end{split}
\end{equation}
For $V\in \mathcal{V}_T$, steps analogous to those leading to~\eqref{stack V-2} yield
\begin{equation*}
\begin{split}
& (X_{1,T} - V W_{0,T})_i \in \mathcal{\hat D}_\delta \text{ for all } i \in \nat_T\\
& \!\iff \!
\smat{
\delta d_\textup{l} - \mathrm{\hat D} x(1)\\
\vdots\\
\delta d_\textup{l} - \mathrm{\hat D} x(T)
}
\le
- (W_{0,T}^\top \otimes \mathrm{\hat D}) \ve(V) 
\le 
\smat{
\delta d_\textup{u} - \mathrm{\hat D} x(1)\\
\vdots\\
\delta d_\textup{u} - \mathrm{\hat D} x(T)
}.
\end{split}
\end{equation*}
Then, we can reparametrize the set $\mathcal{V}_T$ in~\eqref{cal V-2} as
\begin{equation*}
\begin{split}
& \hspace*{-2pt}\mathcal{V}_T= \{ v \in \real^{n(n+m)} \colon \\
& \hspace*{3pt}
\smat{
\delta d_\textup{l} - \mathrm{\hat D} x(1)\\
\vdots\\
\delta d_\textup{l} - \mathrm{\hat D} x(T)
}\le - (W_{0,T}^\top \otimes \mathrm{\hat D}) v \le 
\smat{
\delta d_\textup{u} - \mathrm{\hat D} x(1)\\
\vdots\\
\delta d_\textup{u} - \mathrm{\hat D} x(T)
}
\}.
\end{split}
\end{equation*}
Note that the matrix $W_{0,T}^\top \otimes \mathrm{\hat D} \in \real^{T \hat{n}_d \times (n+m)n}$ (since $W_{0,T}^\top \in \real^{T \times (n+m)}$), and 
\begin{equation}
\label{rank of product}
\rank(W_{0,T}^\top \otimes \mathrm{\hat D}) = \rank(W_{0,T} ) \rank( \mathrm{\hat D})
\end{equation}
by the properties of the Kronecker product \cite[Thm.~4.2.15]{horn1991topics}. With these considerations, we are in a position to prove the statement using Lemma~\ref{lemma:bounded polyhedron} and noting that $\mathcal{V}_T$ is nonempty (as shown in the proof of Lemma~\ref{lemma:rid of the dist-3}, see~\eqref{verify assumption farkas}).
If $W_{0,T}$ has full row rank and $\mathrm{\hat D}$ has full column rank, we conclude that: $\rank(W_{0,T}^\top \otimes \mathrm{\hat D}) = (n+m) n$, $W_{0,T}^\top \otimes \mathrm{\hat D}$ has full column rank, and $\mathcal{V}_T$ is bounded by Lemma~\ref{lemma:bounded polyhedron}. Conversely, if $W_{0,T}$ has \emph{not} full row rank or $\mathrm{\hat D}$ has \emph{not} full column rank,
\begin{equation*}
\begin{split}
& \rank(W_{0,T} ) \rank( \mathrm{\hat D}) < (n+m) \rank( \mathrm{\hat D}) \le (n+m)n \\
& \text{ or } \rank(W_{0,T} ) \rank( \mathrm{\hat D}) < \rank(W_{0,T} ) n \le (n+m)n.
\end{split}
\end{equation*}
We conclude that: $\rank(W_{0,T}^\top \otimes \mathrm{\hat D}) < (n+m)n$ from~\eqref{rank of product}, $W_{0,T}^\top \otimes \mathrm{\hat D}$ has not full column rank, and $\mathcal{V}_T$ is \emph{not} bounded by Lemma~\ref{lemma:bounded polyhedron}.

\subsection{Proof of Theorem~\ref{thm:alternative linear program}}
\label{sec:proof thm alternative linear program}

Since $\mathcal{V}_T$ is bounded, each element $V \in \mathcal{V}_T$ can be written as the convex hull of the (matrix) vertices $V^j$ of $\mathcal{V}_T$ with $j=1, \dots, \VVT$.
The existence of $K \in \real^{m \times n}$ satisfying \eqref{rob-inv} is then equivalent to
\begin{equation*}
\begin{split}
& V^j \bmat{I \\ K } x + d \in \mathcal{S} \text{ for all } x \in \mathcal{S}, d \in \mathcal{D}_\delta, j \in \nat_{\VVT}\\
& \iff \bmat{\mathrm{S}V^j \bmat{I\\ K} & \mathrm{S} }\bmat{x\\d} \le \one \text{ for all } \bmat{x\\ d} \text{ such that }\\
& \hspace*{22mm}\bmat{\mathrm{S} & 0 \\ 0 & \mathrm{D}} 
\bmat{x \\ d} \leq
\bmat{\one \\ \delta \one}
\text{ and all } j \in \nat_{\VVT}.
\end{split}
\end{equation*}
By Fact~\ref{fact:farkas} (whose assumption is verified), this is equivalent to the existence of  $K\in \real^{m \times n}$ and $E^j \in \real^{n_s \times (n_s + n_d)}$ with $j \in \nat_{\VVT}$ satisfying \eqref{alternative linear program}.

\section{Numerical example}
\label{sec:example}

\begin{table}
\centering
\begin{tabular}{llllll}
\toprule
Parameter 	& Value		& Parameter	& Value		& Parameter			& Value\\
\midrule
$\gamma_1$ 	& $0.005$	& $h$		& $0.2$		& $d_0$				& $5$\\
$\gamma_2$ 	& $0.01$	& $\bar d$	& $8$		& $d_1$				& $10$\\
$\tau$		& $0.01$	& $\bar v$	& $8.5$		& $v_\textup{M}$	& $22.2$\\
\bottomrule
\end{tabular}
\caption{Values of parameters of Section~\ref{sec:example}.}
\label{tab:par}
\end{table}

We consider a platoon of two vehicles that should remain close to each other while respecting a safety distance, which is a benchmark problem in (data-driven) safe control (see, e.g., \cite{wabersich2018scalable}).
We refer to Table~\ref{tab:par} for all the values of the parameters used for this example.
Each vehicle $i$ (with $i \in \{1,2\}$) is described as
\begin{equation}
\label{vehicles}
\dot s_i = v_i, \quad \dot v_i = - \gamma_i v_i + a_i
\end{equation}
where $s_i$ is the absolute position of the vehicle along the road, $v_i$ is its velocity, $a_i$ is the force acting on it normalized by its mass, and finally $\gamma_i$ is a viscous friction coefficient (normalized by the mass of the vehicle).
For some desired relative distance $\bar d$ between the vehicles (vehicle $1$ is leading, vehicle $2$ is following) and cruise velocity $\bar v$ of the platoon, we consider the state and input variables
\begin{equation*}
\begin{aligned}
x_1 & = s_1 - s_2 - \bar d,	 &	x_2 & = v_1 - \bar v,			& x_3 & = v_2 - \bar v,\\
	& 						 &	u_1 & = a_1 - \gamma_1 \bar v,	& u_2 & = a_2 - \gamma_2 \bar v
\end{aligned}
\end{equation*}
so that \eqref{vehicles} yields the state equation
\begin{equation*}
\dot x =
\bmat{0 & 1 & -1\\
0 & -\gamma_1 & 0\\
0 & 0 & -\gamma_2}
x
+\bmat{0 & 0\\
1 & 0\\
0 & 1}
u
=:
A_\textup{ct} x + B_\textup{ct} u.
\end{equation*}
To apply our discrete-time results, we consider the Euler discretization with sampling time $\tau$ and the disturbance $d$
\begin{equation}
\label{actual}
x^+ = (I + \tau A_\textup{ct}) x + \tau B_\textup{ct} u + d=: A x + B u + d,
\end{equation}
where the eigenvalues of $A$ are $1$, $1-\tau \gamma_1$, $1-\tau \gamma_2$. These matrices will be used directly in the  model-based design, which we use for comparison, or will constitute the underlying model for the generation of data in the data-based design.

The matrices $\mathrm{S}$ and $\mathrm{D}$ of the sets $\mathcal{S}$ and $\mathcal{D}_\delta$ are
\begingroup
\setlength{\arraycolsep}{2pt}
\begin{equation}
\label{matrices S and D}
\mathrm{S} \!:=\!
\bmat{
    0.9165   &  0.1900   & -0.1762\\
    1.4250   &  0.0661   & -0.0769\\
   -0.0322   &  0.1925   &  0.2165\\
   -0.9165   & -0.1900   &  0.1762\\
   -1.4250   & -0.0661   &  0.0769\\
    0.0322   & -0.1925   & -0.2165\\}\!,
\mathrm{D}\!:=\!
\begin{bmatrix}
     1   &  0   &  0\\
    -1   &  0   &  0\\
     0   &  1   &  0\\
     0   & -1   &  0\\
     0   &  0   &  1\\
     0   &  0   & -1\\
\end{bmatrix}.
\end{equation}
\endgroup
and ensure by~\cite{gravalou1994algorithm} that \eqref{MB-robInv-2} has a feasible model-based solution, to which we can compare our data-based solution. With this $\mathrm{D}$, we explore different values of $\delta$ for the set $\mathcal{D}_\delta$ in the following.
The considered linear programs are solved using YALMIP \cite{lofberg2004yalmip} and MOSEK.

We first show that a feasible model-based solution can be found according to~\eqref{MB-robInv-2}.
When $A$ and $B$ in~\eqref{actual} are known, \eqref{MB-robInv-2} can be solved in the decision variables $E$ and $K$ for the sets $\mathcal{S}$ and $\mathcal{D}_\delta$ as in~\eqref{matrices S and D} and with $\delta$ up to approximately $0.0625$. For that $\delta$, we obtain
\begin{equation*}
K= \bmat{
 -351.5550 &	-102.1629 	& 13.4292\\
  370.5020 &	-11.3050	& -71.7106\\
}.
\end{equation*}
The evolution of the closed-loop system made of \eqref{actual} and $u=K x$ is in Figure~\ref{fig:MB}.
We note that in this figure and the following analogous ones, the disturbance $d$ in~\eqref{actual} is always taken randomly on the vertices of $\mathcal{D}_\delta$, corresponding to a worst case.

\begin{figure}
\centerline{\includegraphics[width=.54\textwidth]{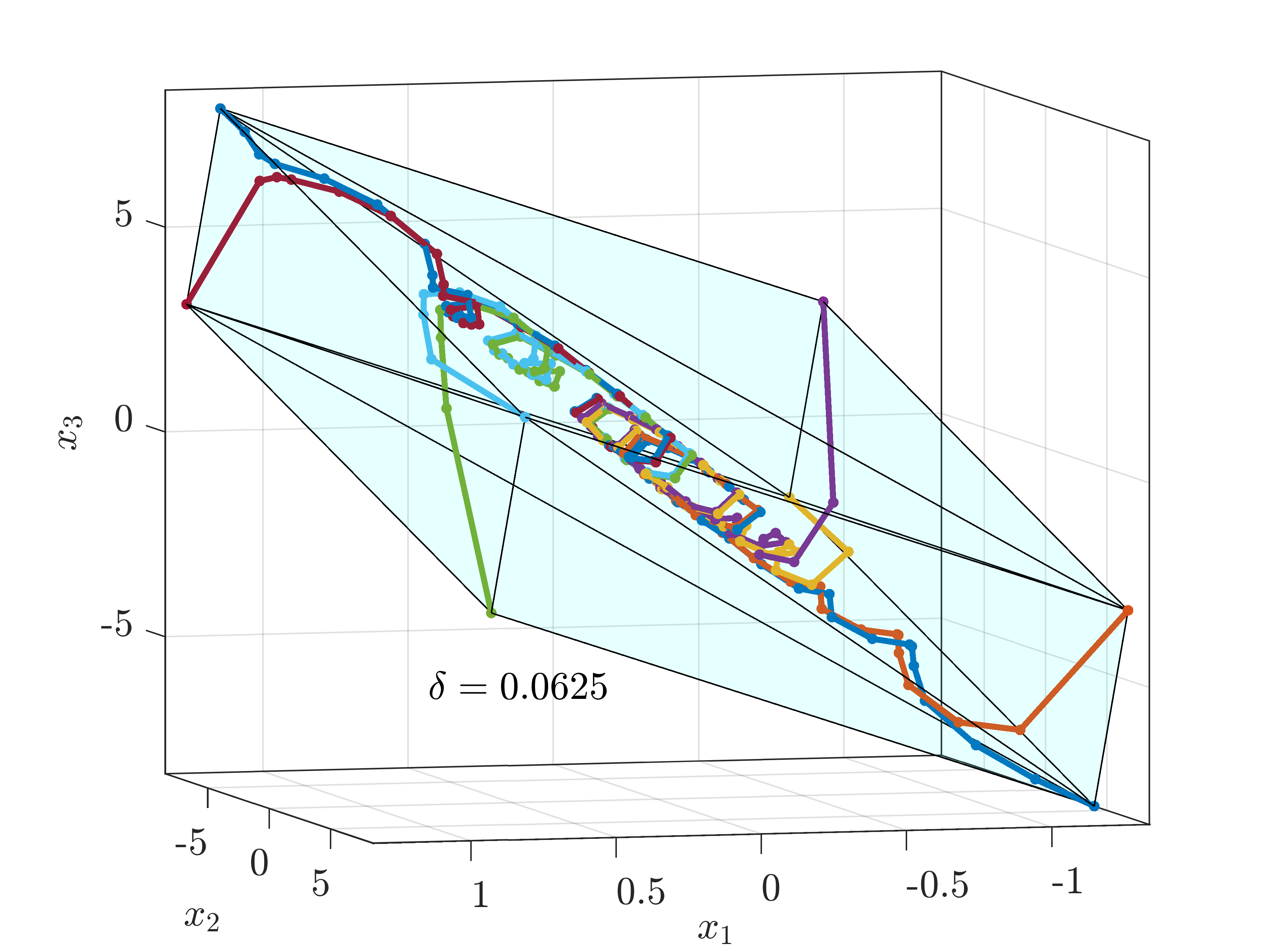}}
\caption{Solutions obtained for the closed loop of \eqref{actual} with the gain tuned
by the model-based approach. The cyan surface denotes the boundary of the polyhedron $\mathcal{S}$.}
\label{fig:MB}
\end{figure}

We show then our data-based solution and use a large number of data, i.e., $T=1600$ samples, as in Figure~\ref{fig:data}.
The input data $u(k)$ ($k=0, \dots, T-1$) are selected by generating each of the $m=2$ components according to a random variable uniformly distributed in $[-5,5]$.
The disturbance samples $d(k)$ ($k=0, \dots, T-1$) are selected by generating each of the $n=3$ components according to a random variable uniformly distributed in $[-\delta,\delta]=[-0.05,0.05]$ (consistently with $\mathrm{D}$ in~\eqref{matrices S and D}).
On one hand, the required number of data can be lower (in particular, a necessary condition to have ``rich'' data as in Section~\ref{sec:benefits large and rich} is $T \ge n+m=5$).
On the other hand, we choose a large value of $T$ to highlight that the solution of the considered linear program is computationally viable even with a large amount of samples and attains $\delta$ very close to the model-based solution.

\begin{figure}
\centerline{\includegraphics[width=.49\textwidth]{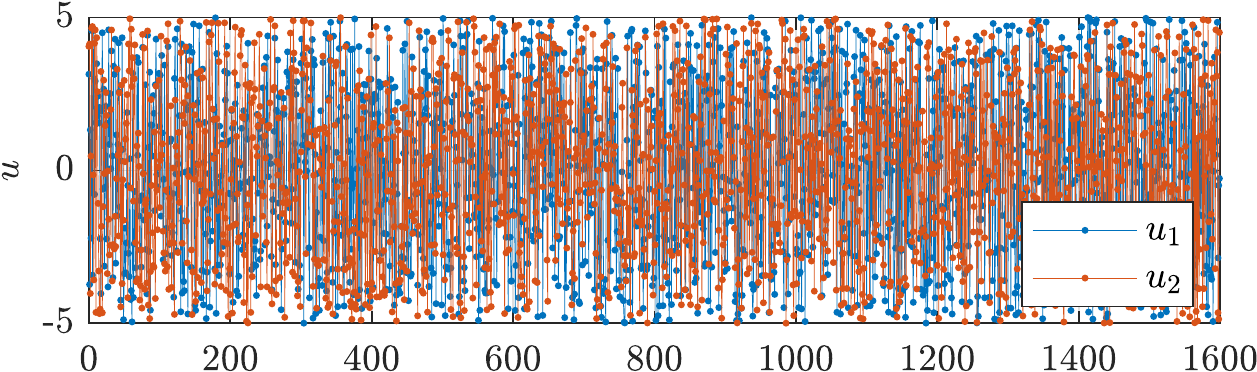}}
\centerline{\includegraphics[width=.49\textwidth]{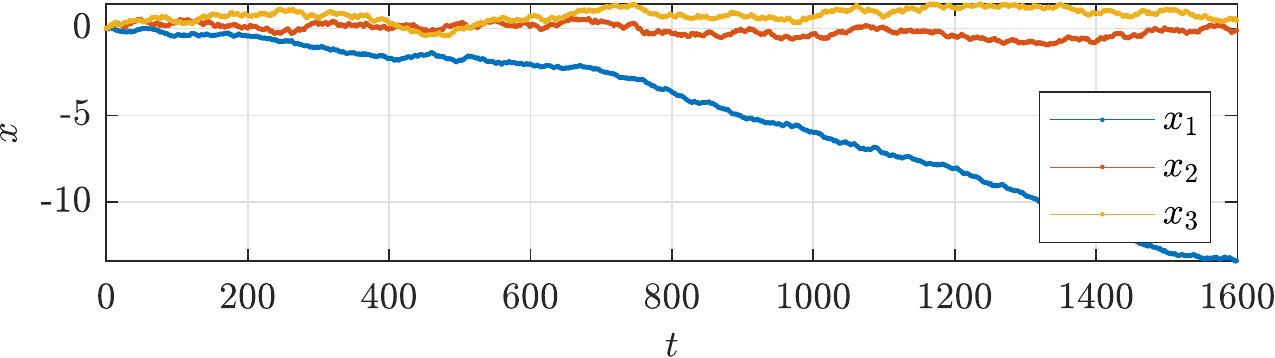}}
\caption{The considered input and state data with for $T=1600$.}
\label{fig:data}
\end{figure}

For these data organized in the matrices $U_{0,T}$, $X_{0,T}$, $X_{1,T}$ in~\eqref{U0}-\eqref{X1}, we solve the linear program in~\eqref{main:without x and dist BIS}, which depends only on noisy data, and obtain
\begin{equation*}
K= \bmat{
 -145.7158  & -30.7420  &  13.4551\\
  164.2787  &  11.8638  & -29.0357\\
}
\end{equation*}
and the closed-loop evolutions in Figure~\ref{fig:DB}.

\begin{figure}
\centerline{\includegraphics[width=.54\textwidth]{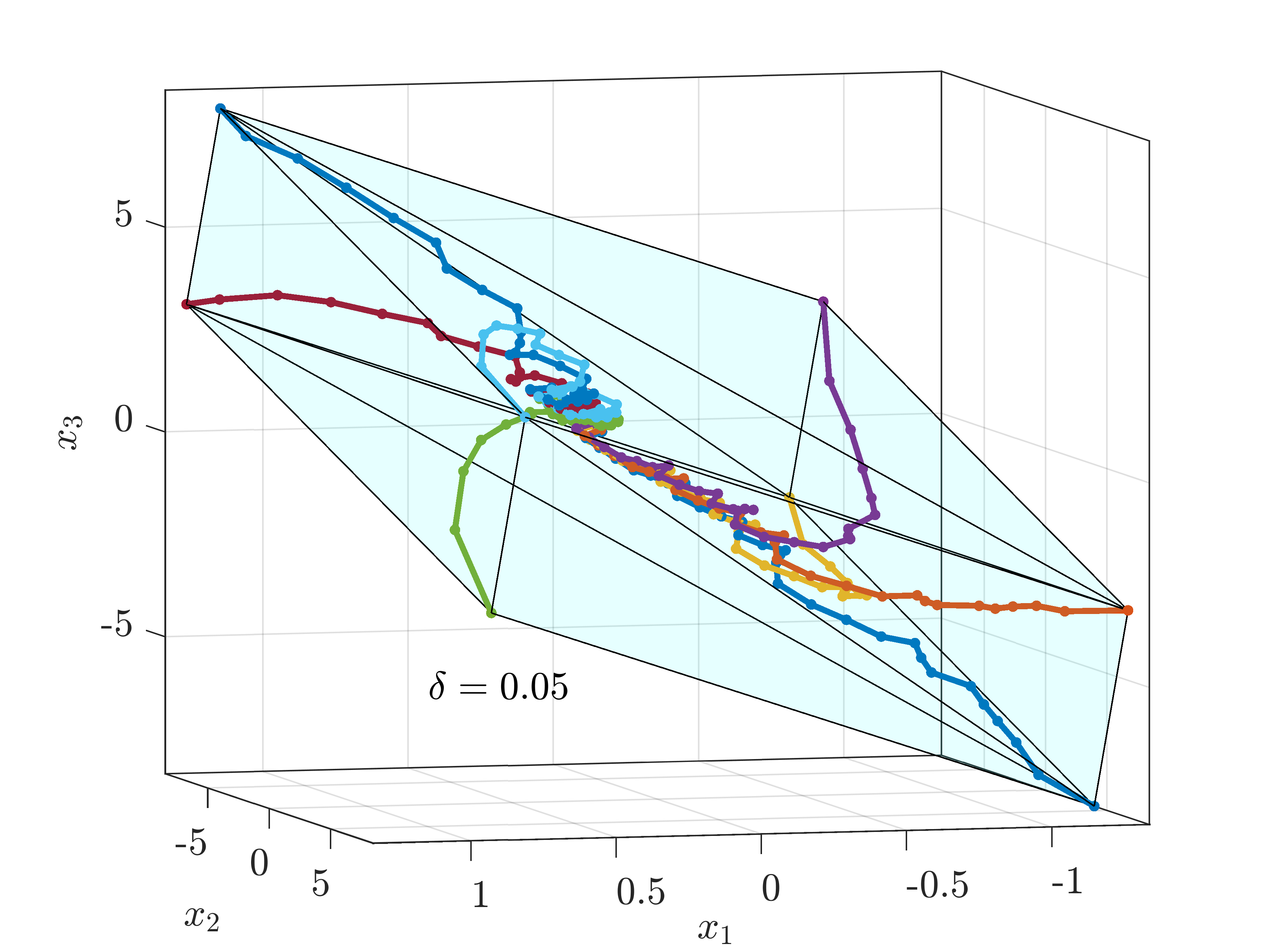}}
\caption{Solutions obtained for the closed loop of \eqref{actual} with the gain tuned
by the data-based approach.}
\label{fig:DB}
\end{figure}

Finally, we characterize in Figure~\ref{fig:deltaT} how different values of $T$ can obtain feasible data-based solutions. Data-based solutions are compared against the maximum $\delta=0.0625$ obtained by the model-based solution, whereby we note that values above $0.0625$ led to infeasible solutions both for the model-based and data-based solution.
The larger $T$, the more information is carried by the data on the system dynamics, and this makes it easier to satisfy the data-based feasibility problem, as observed in Section~\ref{sec:benefits large and rich}. 
The data-based solution compares tightly with the model-based solution for large $T$, such as $\delta=0.0575$ for $T=3000$. 
At the same time, with smaller values of $T$ such as $600$, the data-based solution attains $\delta=0.0275$ compared to the maximum $\delta=0.0625$.

\begin{figure}
\centerline{\includegraphics[width=.49\textwidth]{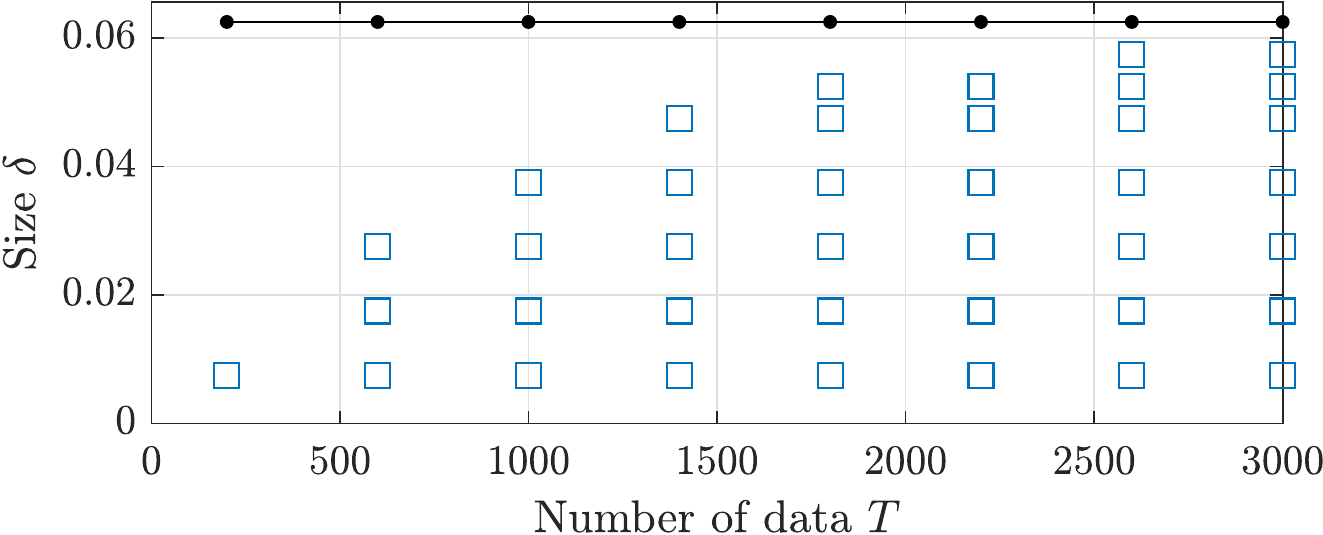}}
\caption{Feasibility analysis for $\delta$ and $T$. The squares correspond to pairs $(T,\delta)$ for which a feasible data-based solution was found by the solver. The black line denotes the maximum $\delta$ of feasible model-based solutions.}
\label{fig:deltaT}
\end{figure}

\section{Conclusion}

For a discrete-time linear system perturbed by a process disturbance, we obtained necessary and sufficient conditions for robust invariance in terms of noisy input-state data, given by a linear program.

The corresponding explicit formulae can be used, with some conservatism, to treat nonlinear systems when their state and input are constrained on compact sets \cite{aswani2013provably}. More specific topics of future work are (i)~a more quantitative characterization of the set $\mathcal{V}_T$ as the number $T$ of data varies, (ii)~the consideration of input-output data and output feedback, (iii) designing the data-collection experiment to explicitly account for the state constraints (according to the paradigm of safe exploration),  
(iv) relaxing the linear control policy $u=K x$ to piecewise-linear or nonlinear ones (as the former can be conservative for a given set $\mathcal{S}$) through techniques from the model-based case \cite{gutman1986admissible,gilbert1991linear} or new ones tailored on data.

\section*{Acknowledgment}

We would like to thank Professor Mario Sznaier for pointing out to us the related work in~\cite{Dai2019cdc,Dai2020aut,dai2020semialgebraic,Dai2020ifac}.

\appendix

This appendix gives the proof of Fact~\ref{fact:farkas}, which relies on the nonhomogeneous Farkas theorem \cite[p.~32]{mangasarian1994nonlinear} reported next.
\begin{fact}\emph{\cite[p.~32]{mangasarian1994nonlinear}}
For each $\sa{i} \in \{1, \dots, \sa{q}\}$, let $\sa{A} \in \real^{\sa{p} \times \sa{n}}$, $\sa{b}_\sa{i} \in \real^{\sa{n}}$, $\sa{c} \in \real^{\sa{p}}$, $\sa{d}_\sa{i} \in \real$.
\eqref{alternative-1} and \eqref{alternative-2} are equivalent:
\begin{align}
\label{alternative-1}
& \!\!\! \text{there does \emph{not} exist } \sa{x} \! \in \! \real^\sa{n}\! \text{ such that }\! (\sa{b}_\sa{i}^\top \sa{x}  > \sa{d}_\sa{i}, \sa{A} \sa{x} \le \sa{c});\!\\
&
\!\!\! \text{there exists } \sa{e}_\sa{i} \in \real^\sa{p} \text{ such that } (\sa{A}^\top \sa{e}_\sa{i} = \sa{b}_\sa{i},\sa{c}^\top \sa{e}_\sa{i} \le \sa{d}_\sa{i}, \sa{e}_\sa{i} \ge 0) \notag \\
& \hspace*{5mm}\text{or } (\sa{A}^\top \sa{e}_\sa{i} = 0,\sa{c}^\top \sa{e}_\sa{i} < 0,\sa{e}_\sa{i} \ge 0).
\label{alternative-2}
\end{align}
\end{fact}
By logical equivalences and transpositions, \eqref{alternative-1} and \eqref{alternative-2} are respectively equivalent to
\begin{align}
& \sa{b}_\sa{i}^\top \sa{x}  \le \sa{d}_\sa{i} \text{ for all }  \sa{x} \in \real^\sa{n} \text{ such that } \sa{A} \sa{x} \le \sa{c};
\label{alternative-1-prime}
\\
&
\text{there exists } \sa{e}_\sa{i} \in \real^\sa{p} \text{ such that } (\sa{e}_\sa{i}^\top \sa{A} = \sa{b}_\sa{i}^\top, \sa{e}_\sa{i}^\top \sa{c} \le \sa{d}_\sa{i}, \sa{e}_\sa{i}^\top \ge 0) \notag \\
& \hspace*{5mm}\text{or } (\sa{e}_\sa{i}^\top \sa{A} = 0,\sa{e}_\sa{i}^\top \sa{c} < 0,\sa{e}_\sa{i}^\top \ge 0).
\label{alternative-2-prime}
\end{align}
Under the assumption that there exists $\sa{z} \in \real^\sa{n}$ satisfying $\sa{A} \sa{z} \le \sa{c}$, the case
\begin{equation}
\label{ruled-out-case}
\sa{e}_\sa{i}^\top \sa{A} = 0,\sa{e}_\sa{i}^\top \sa{c} < 0,\sa{e}_\sa{i}^\top \ge 0
\end{equation}
in~\eqref{alternative-2-prime} can never occur. By contradiction, suppose there exists $\sa{e}_\sa{i} \in \real^\sa{p}$ satisfying \eqref{ruled-out-case}. Since $\sa{e}_\sa{i}^\top \ge 0$, $\sa{A} \sa{z} \le \sa{c}$ would imply that $\sa{e}_\sa{i}^\top \sa{A} \sa{z} \le \sa{e}_\sa{i}^\top \sa{c}$, i.e., $0 \le \sa{e}_\sa{i}^\top \sa{c}$ (since $\sa{e}_\sa{i}^\top \sa{A} = 0$), thereby contradicting $\sa{e}_\sa{i}^\top \sa{c} < 0$. Hence, under the assumption that there exists $\sa{z} \in \real^\sa{n}$ satisfying $\sa{A} \sa{z} \le \sa{c}$, \eqref{alternative-1-final} and \eqref{alternative-2-final} are equivalent:
\begin{align}
& \sa{b}_\sa{i}^\top \sa{x}  \le \sa{d}_\sa{i} \text{ for all }  \sa{x} \in \real^\sa{n} \text{ such that } \sa{A} \sa{x} \le \sa{c};
\label{alternative-1-final}
\\
&
\text{there exists } \sa{e}_\sa{i} \in \real^\sa{p} \notag \\
& \hspace*{1.5cm} \text{ such that } (\sa{e}_\sa{i}^\top \sa{A} = \sa{b}_\sa{i}^\top, \sa{e}_\sa{i}^\top \sa{c} \le \sa{d}_\sa{i}, \sa{e}_\sa{i}^\top \ge 0).
\label{alternative-2-final}
\end{align}
The equivalence of \eqref{alternative-1-final} and \eqref{alternative-2-final} proves the statement of Fact~\ref{fact:farkas} by considering
\begin{equation*}
\sa{B}=\smat{\sa{b}_1^\top\\ \vdots \\ \sa{b}_\sa{q}^\top}, \sa{d}=\smat{\sa{d}_1 \\ \vdots\\ \sa{d}_\sa{q}}, \sa{E}=\smat{\sa{e}_1^\top\\ \vdots \\ \sa{e}_\sa{q}^\top}.
\end{equation*}

\bibliographystyle{plain}
\bibliography{references}

\end{document}